\documentclass[twocolumn,pra,aps]{revtex4-1}

\usepackage[dvipdfmx]{graphicx}
\usepackage{color, fancybox} 
\usepackage{amsfonts}
\usepackage{amssymb,amsmath,amsthm} 
\usepackage{braket}

\begin{document}
\newcommand{\mbold}[1]{\mbox{\boldmath $#1$}}
\newcommand{\sbold}[1]{\mbox{\boldmath ${\scriptstyle #1}$}}
\newcommand{\tr}{\,{\rm tr}\,}
\renewcommand{\arraystretch}{1.5}
\newtheorem*{theorem}{Theorem}
\title{Quantum-state comparison and discrimination}
\author{A.~Hayashi, T.~Hashimoto, and M.~Horibe}
\affiliation{Department of Applied Physics, 
          University of Fukui, Fukui 910-8507, Japan}
\begin{abstract}
We investigate the performance of discrimination strategy in the comparison task of 
known quantum states.
In the discrimination strategy, one infers whether or not two quantum systems are in the 
same state on the basis of the outcomes of separate discrimination measurements on 
each system.  In some cases with more than two possible states, the optimal strategy in 
minimum-error comparison is that one should infer the two systems are in different 
states without any measurement, implying that the discrimination strategy performs worse 
than the trivial ``no-measurement''  strategy. We present a sufficient condition for 
this phenomenon to happen. 
For two pure states with equal prior probabilities, we determine the optimal comparison 
success probability with an error margin, which interpolates the minimum-error and 
unambiguous comparison. We find that the discrimination strategy is not optimal except 
for the minimum-error case.  
\end{abstract}

\maketitle

\section{Introduction} \label{sec_introduction}
The laws of quantum mechanics do not allow one to distinguish nonorthogonal 
quantum states perfectly \cite{Helstrom76, Holevo82, Chefles00, Nielsen_text_book}.  
First, this is because of the statistical nature of 
quantum measurement, which generally destroys the state of the system, and, 
further, one cannot clone an unknown quantum states \cite{Wootters82}. 

Quantum state comparison is one of the problems which are directly related to this 
nature of quantum mechanics 
\cite{Barnett03, Klenmann05,  Chefles04, Filippov12, Pang11, Andersson06,Olivares10}. Suppose we are given two quantum systems, and the task is 
to optimally infer whether or not the two systems are in the same state. 
We can consider two different settings of the problem. 
One is the case in which the possible states are unknown; 
that is, we have no classical knowledge on the states. The other is the case in which 
the state is selected from a known set of states with some prior probabilities. 
We concentrate on the latter case.  

Suppose two states are independently selected from a set of two known pure states.  
In Ref.~\cite{Barnett03}, Barnett {\it et al}.  showed that the optimal comparison 
in the minimum-error scheme is attained by the discrimination strategy, by which we 
mean we separately perform the optimal discrimination 
measurement on each system, and if the two outcomes are equal we 
infer that the two systems are in the same state, and otherwise they are 
in different states. On the other hand, they showed that the optimal comparison 
in the unambiguous scheme requires a collective measurement on the whole system, 
and cannot be attained by the unambiguous discrimination 
\cite{Ivanovic87, Dieks88, Peres88, Jaeger95} strategy.  
This conclusion was subsequently generalized to the unambiguous comparison of two pure 
states with arbitrary prior probabilities  \cite{Klenmann05}. 

In this paper, we further investigate the performance of the discrimination strategy 
in the comparison task of known quantum states. First we show that 
the discrimination strategy is optimal in general two-state minimum-error comparison; 
that is, for generally mixed two states with arbitrary prior probabilities (Sec.~\ref{sec_two_state}).  Then we study the comparison task involving more than two pure 
states (Sec.~\ref{sec_many}). In some cases, we find that the optimal strategy is trivial; 
that is, we should simply infer that two systems are in different states without performing 
any measurement (``no-measurement'' strategy). In these cases, the discrimination 
strategy performs worse than the ``no-measurement'' strategy, and is a waste of effort.  
The condition for this phenomenon to occur is also discussed. 
 
As stated previously, for a two-state comparison, the discrimination strategy is not optimal 
in the unambiguous scheme, whereas it is optimal in the minimum-error setting.   
We determine the optimal comparison success probability with an error margin 
for two pure states with equal prior probabilities (Sec.~\ref{sec_error_margin}). 
This error-margin scheme interpolates the minimum-error and unambiguous settings 
\cite{Touzel07,Hayashi08,Sugimoto09}. 
We find that the discrimination strategy is optimal only in the minimum-error case; that is, 
the optimal comparison requires collective measurement on the whole system as 
long as the error-margin condition is active. 
 
\section{Two-state comparison and discrimination strategy} 
\label{sec_two_state}
Suppose we are given two quantum systems, which are independently
prepared in one of two known pure states $\ket{\phi_1}$ and $\ket{\phi_2}$ 
with equal prior probabilities.  The state of the combined system is thus 
either one of $\ket{\phi_1,\phi_1}$, $\ket{\phi_1,\phi_2}$,
$\ket{\phi_2,\phi_1}$, or $\ket{\phi_2,\phi_2}$ with probabilities 
$1/4$. The task is to infer whether the two systems are in the same state or not. 
This problem has been addressed and solved by Barnett {\it et al.} in Ref.~\cite{Barnett03}. 
In this section, we first reproduce their results, and then extend the results to a more 
general case. 

Measurement is described by a positive operator-valued measure (POVM), 
$\{E_=, E_{\ne}\}$, where measurement outcome ``$=$'' corresponds to the 
guess that two systems are in the same state whereas outcome ``$\ne$'' 
means their states are different. The probability of success in this state 
comparison is given by
\begin{align}
  P_\circ =& \frac{1}{4}\sum_{k=1}^2 \bra{\phi_k, \phi_k}E_=\ket{\phi_k, \phi_k}
                                  \nonumber \\ 
     &+\frac{1}{4}\sum_{k,j=1(k \ne j)}^{2} 
                                          \bra{\phi_k, \phi_j}E_{\ne} \ket{\phi_j, \phi_k}.
                         \label{eq_success}
\end{align}
In this section we adopt the minimum error scheme; we maximize the success 
probability $P_\circ$ without any constraint on the probability of an erroneous guess. 
Using $E_{\ne} = \mbold{1}-E_=$, we write the success probability $P_\circ$ as
\begin{align}
  P_\circ = \frac{1}{2} + \tr E_= \Lambda,
\end{align}
where 
\begin{align}
 \Lambda =& \frac{1}{4} \sum_{k=1}^2 
             \ket{\phi_k}\bra{\phi_k}\otimes \ket{\phi_k}\bra{\phi_k}
                                   \nonumber \\
                   &- \frac{1}{4} \sum_{k,j=1(k \ne j)}^2 
             \ket{\phi_k}\bra{\phi_k} \otimes  \ket{\phi_j}\bra{\phi_j}.
\end{align}
Since $0 \le E_= \le\mbold{1}$, the maximal value of $\tr E_= \Lambda$ is 
given by the sum of all positive eigenvalues of the operator $\Lambda$.  

To obtain eigenvalues of $\Lambda$, it suffices to work in the two-dimensional 
space spanned by the states  $\ket{\phi_1}$ and $\ket{\phi_2}$.   
Introducing Bloch vectors $\mbold{n}_k$ for the states $\ket{\phi_k}$, we have  
\begin{align}
   \ket{\phi_k}\bra{\phi_k} = 
   \frac{ \mbold{1} + \mbold{n}_k \cdot \mbold{\sigma} }{2}, 
                              \label{eq_Bloch}
\end{align}
where $\mbold{\sigma} = (\sigma_x,\sigma_y,\sigma_z)$ are the Pauli 
matrices. With this relation the operator $\Lambda$ takes the form
\begin{align}
    \Lambda = \frac{1}{16}(\mbold{n}_1-\mbold{n}_2) \cdot \mbold{\sigma}
                        \otimes (\mbold{n}_1-\mbold{n}_2) \cdot \mbold{\sigma}. 
\end{align}
It is now easy to obtain the eigenvalues of $\Lambda$, since the eigenvalues of 
$(\mbold{n}_1-\mbold{n}_2) \cdot \mbold{\sigma}$ are given by 
$\pm |\mbold{n}_1-\mbold{n}_2|$. Thus we find the optimal success probability 
is given by 
\begin{align}
  P^{\rm opt}_\circ = \frac{1}{2}+\frac{1}{8} |\mbold{n}_1-\mbold{n}_2|^2
                                 = 1-\frac{1}{2}|\braket{\phi_1|\phi_2}|^2, 
                     \label{eq_minimumerror}
\end{align} 
which is attained when $E_=$ is the projector onto the subspace spanned by the 
eigenstates of $\Lambda$ with positive eigenvalues.

Let us examine the obtained results more closely. The optimal comparison is realized 
when $E_=$ is given by
\begin{align}
    E_= = e_1 \otimes e_1 + e_2 \otimes e_2,  \label{eq_disc_str}
\end{align}
where $e_1$ and $e_2$ are the projectors onto the eigenspaces of 
$(\mbold{n}_1-\mbold{n}_2) \cdot \mbold{\sigma}$ with positive and negative 
eigenvalues, respectively.  
Now recall the minimal-error discrimination problem
between the two pure states $\ket{\phi_1}$ and $\ket{\phi_2}$ with equal 
prior probabilities. The optimal measurement in this discrimination problem is given 
by the POVM $\{e_1,e_2\}$, where $e_k$ corresponds to the guess that the state is 
$\ket{\phi_k}$.  
This implies that the optimal state comparison under consideration is reduced to
the optimal discrimination; we separately perform the optimal discrimination 
measurement on each system of the two, and if the two outcomes are equal we 
infer that the two systems are in the same state, and otherwise they are 
in different states. 

We can show that this conclusion holds in more general case: a state comparison of 
two mixed states $\rho_1$ and $\rho_2$ with arbitrary prior probabilities 
$\eta_1$ and $\eta_2$, respectively. In this case the success probability of 
comparison is written as
\begin{align}
  P_\circ =& \tr E_= \left(
             \eta_1^2 \rho_1 \otimes \rho_1 + \eta_2^2 \rho_2 \otimes \rho_2 
                                 \right)                        \nonumber \\
                & + \tr E_{\ne} \left(
             \eta_1\eta_2  \rho_1 \otimes \rho_2 + \eta_2\eta_1 \rho_2 \otimes \rho_1
                                          \right)               \nonumber \\
               =& 2\eta_1\eta_2 + \tr E_= 
                       (\eta_1\rho_1-\eta_2\rho_2) \otimes  (\eta_1\rho_1-\eta_2\rho_2). 
\end{align}
The optimal POVM element $E_{=}$ is clearly given in the form of Eq.~(\ref{eq_disc_str}) 
with $e_1$ ($e_2$) being the projector onto the eigenspace of 
$\eta_1\rho_1-\eta_2\rho_2$ with positive (negative) eigenvalues. 
This POVM $\{e_1,e_2\}$ is the optimal POVM in the discrimination problem between 
$\rho_1$ and $\rho_2$ with prior probabilities $\eta_1$ and $\eta_2$, respectively. 
This is evident since the success probability $Q_\circ$  of this discrimination problem is written as 
\begin{align}
   Q_\circ &= \eta_1\tr e_1\rho_1 + \eta_2\tr e_2 \rho_2
                                                         \nonumber \\
               &= \eta_2  + \tr e_1(\eta_1\rho_1-\eta_2\rho_2). 
\end{align} 
In conclusion, the discrimination strategy is optimal in the general two-state 
minimum-error comparison.  
\section{Many-states comparison}
\label{sec_many} 
In the preceding section we have shown that the discrimination strategy is optimal in the 
general two-state minimum-error comparison. It is interesting whether this result holds in 
the minimum-error comparison of more than two states. 
To investigate this issue, we take an example of minimum-error comparison involving 
$N \ge 2$ states.   

In a two dimensional space,  we consider the following $N$ pure states:
\begin{align}
    \ket{\phi_k} = U_k \ket{\phi}, (k=0,1,\cdots,N-1), 
\end{align}
where the initial state $\ket{\phi}$ is given by 
\begin{align}
  \ket{\phi} = \frac{\ket{0}+\ket{1}}{\sqrt{2}},
\end{align}
and the phase shift operator $U_k$ is defined as  
\begin{align}
  \left\{ \begin{array}{ll}
                   U_k \ket{0} &= \ket{0}  \\
                   U_k \ket{1} &= e^{i\frac{2\pi}{N}k}\ket{1}  \\
              \end{array}
  \right. .
\end{align}
Note that $\{U_k\}_{k=0}^{N-1}$ is a unitary representation of $\mathbb{Z}_N$, 
consisting of two inequivalent one-dimensional irreducible representations. One 
irreducible representation space is spanned by $\ket{0}$, and the other by $\ket{1}$. 
The discrimination problem of these $N$ states with equal prior probabilities has been  
analyzed in Ref. \cite{Hashimoto10}.  The optimal minimum-error discrimination success 
probability $Q^{\rm opt}_{\circ}$ is found to be 
\begin{align}
    Q^{\rm opt}_{\circ} = \frac{2}{N},
\end{align}
with the optimal POVM $e_k^{\rm opt}$, corresponding to the guess that the state is 
$\ket{\phi_k}$, given by
\begin{align}
   e_k^{\rm opt} = \frac{2}{N} \ket{\phi_k}\bra{\phi_k}.
\end{align}

Now we consider the comparison problem of the $N$ states $\ket{\phi_k}$; the two systems 
are prepared in one of those states independently with equal prior probabilities, and the 
task is to guess whether the two systems are in the same state or not.   As in the preceding 
section, we write the comparison success probability $P_{\circ}$ as 
\begin{align}
  P_{\circ} = 1 - \frac{1}{N} + \tr E_{=}\Lambda,
\end{align}
where  
\begin{align}
    \Lambda =& \frac{1}{N^2} \sum_{k=0}^{N-1} 
             \ket{\phi_k}\bra{\phi_k}\otimes \ket{\phi_k}\bra{\phi_k}
                                   \nonumber \\
                   &- \frac{1}{N^2} \sum_{k,j=0(k \ne j)}^{N-1} 
             \ket{\phi_k}\bra{\phi_k} \otimes  \ket{\phi_j}\bra{\phi_j}.
                             \label{eq_Lambda}
\end{align}
The optimal POVM $E_{=}^{\rm opt}$ is the projector onto the eigenspace of $\Lambda$ 
with positive eigenvalues. 
In order to calculate $\Lambda$, we again use the relation 
Eq.~(\ref{eq_Bloch}) where the Bloch vector for state $\ket{\phi_k}$ is now given by 
\begin{align}
    \mbold{n}_k &= \left( n_k^x, n_k^y, n_k^z \right) \nonumber \\ 
       &= \left( \cos \frac{2\pi}{N}k, \sin \frac{2\pi}{N}k, 0 \right). 
\end{align}
The following formulas of Bloch vectors are useful for performing the summation over 
$k$ and $j$ in the expression of $\Lambda$:  
\begin{align}
  & \sum_{k=0}^{N-1} n_k^x= 
              \sum_{k=0}^{N-1} n_k^y  = 0 , \\
  & \sum_{k=0}^{N-1} n_k^x n_k^x =
      \left\{ \begin{array}{ll}
                         2   ,  & ( N =2 )  \\
                          N/2      ,  & ( N \ge 3) \\
                  \end{array}
      \right.,   \\
  & \sum_{k=0}^{N-1} n_k^y n_k^y =
      \left\{ \begin{array}{ll}
                         0   ,  & ( N =2 )  \\
                          N/2      ,  & ( N \ge 3) \\
                  \end{array}
      \right.,   \\
   &  \sum_{k=0}^{N-1} n_k^x n_k^y = 0. 
\end{align}
We find that $\Lambda$ is expressed as follows: 
\begin{align}
   \Lambda = \left\{ 
      \begin{array}{l} 
        \frac{1}{4}\sigma_x \otimes \sigma_x,\   (N=2) \\
        \frac{1}{4N} (\sigma_x \otimes \sigma_x + \sigma_y \otimes \sigma_y)
       + \frac{1}{2N}-\frac{1}{4},\   (N \ge 3) \\
      \end{array}
                     \right. .
\end{align}

When $N=2$, the problem is trivial since the states $\ket{\phi_0}$ and $\ket{\phi_1}$ 
are orthogonal and can be perfectly discriminated. In fact  
the sum of positive eigenvalues of $\Lambda$ is 1/2, and we 
obtain $P_{\circ}^{\rm opt}=1$ with the optimal POVM given by the 
discrimination strategy 
\begin{align}
     E^{\rm opt}_{=} = E^{\rm disc}_{=} \equiv 
            \sum_{k=0}^1 e_k^{\rm opt} \otimes  e_k^{\rm opt}. 
\end{align}

Let us see the eigenvalues of $\Lambda$ when $N \ge 3$. 
The four eigenstates of $\Lambda$ are the Bell states:
\begin{align}
    \left\{ 
    \begin{array}{ll}
       \ket{\Psi_{\pm}} &=  \frac{1}{\sqrt{2}}\left( \ket{00}\pm \ket{11} \right)  \\
       \ket{\Phi_{\pm}} &= \frac{1}{\sqrt{2}}\left( \ket{01}\pm \ket{10} \right)  \\
    \end{array}
    \right.. 
\end{align}
We find that the eigenvalues associated with $\ket{\Psi_{\pm}}$ and $\ket{\Phi_{-}}$ 
are all negative, and the eigenvalue of $\ket{\Phi_{+}}$  is given by 
\begin{align}
   \lambda_{\Phi_+} = \frac{1}{N}-\frac{1}{4} =
    \left\{ \begin{array}{l}
                    > 0,\ (N=3) \\
                    \le 0, \ (N \ge 4) \\
                \end{array}
     \right. .
\end{align}

Therefore, in the case of $N=3$, the optimal comparison success probability is given by 
 $P_{\circ}^{\rm opt} = 3/4$ with the optimal POVM 
$E^{\rm opt}_{=} = \ket{\Phi_+}\bra{\Phi_+}$.
If we take the discrimination strategy, 
$E_{=}^{\rm disc} = \sum_{k=0}^2 e_k^{\rm opt} \otimes e_k^{\rm opt}$, 
we find $\tr E_{=}^{\rm disc} \Lambda$  vanishes. This implies that 
$P_{\circ}^{\rm disc} = 2/3 $, which is strictly less than $P_{\circ}^{\rm opt} = 3/4$. 
Thus the state comparison is not reduced to the discrimination problem in this case.  
It should be noted that the success probability 2/3 of the discrimination strategy can 
be obtained by a simpler strategy, where $E_{=}$ is set to zero and 
therefore $E_{\ne}=\mbold{1}$; namely, 
we always infer the two systems are in different states without performing any measurement 
(no-measurement strategy). 

When $N \ge 4$, no eigenvalue of $\Lambda$ is positive.  
Therefore the optimal comparison success probability is given by 
\begin{align}
 P_{\circ}^{\rm opt} = 1-1/N, 
\end{align}
with $E_{=}^{\rm opt} = 0$. 
Thus, rather surprisingly, we find  that the optimal strategy is just the no-measurement strategy. 
The discrimination strategy gives a worse result since 
$\tr E_{=}^{\rm disc} \Lambda$ is negative.  
The size relation of the success probabilities by the three strategies is summarized 
in Table \ref{tab_P}.

In general, we expect that the no-measurement strategy performs better when the 
number of possible states is large, since the probability of selecting two different states 
dominates.  In what follows, we present a sufficient condition for the no-measurement 
strategy to be optimal. For general $N$ pure states 
$\ket{\phi_k}(k=0,\ldots,N-1)$ with equal prior probabilities, we write the 
$\Lambda$ operator corresponding to Eq.~(\ref{eq_Lambda}) in the following form:
\begin{align}
  \Lambda = \frac{2}{N}R -r \otimes r,
\end{align}
where we introduced two normalized density operators $R$ and $r$ defined as  
\begin{align}
     R & \equiv \frac{1}{N} \sum_{k=0}^{N-1}\ket{\phi_k,\phi_k}\bra{\phi_k,\phi_k}, \\
     r & \equiv  \frac{1}{N}\sum_{k=0}^{N-1} \ket{\phi_k}\bra{\phi_k}.
\end{align}
Let $\lambda_{\max}$ and $\lambda_{\min}$ be the maximum and minimum of all 
nonzero eigenvalues of $r$, respectively. Note that $\lambda_{\min} > 0$. 
We show that if 
\begin{align} 
    \lambda_{\min} \ge \sqrt{\frac{2\lambda_{\max}}{N}},  \label{eq_sufficient}
\end{align} 
then $\Lambda$ is negative semidefinite, and consequently the no-measurement 
strategy is optimal.  
Let $V_r$ be the support of $r$. Clearly, the support of $r \otimes r$ is given by 
$V_r^{\otimes 2}$, and the support of $R$ is a subspace of 
$V_r^{\otimes 2}$. Now take an arbitrary vector $\ket{\Phi}$ in the total space 
and express it as $\ket{\Phi}=\ket{\Phi_\parallel}+\ket{\Phi_\perp}$, where
$\ket{\Phi_\parallel} \in V_r^{\otimes 2}$ and $\ket{\Phi_\perp}$ is the 
component perpendicular to $V_r^{\otimes 2}$.  
Assume the condition  (\ref{eq_sufficient}) holds. Then we observe 
\begin{align}
   \braket{\Phi | \Lambda |\Phi} &= 
           \frac{2}{N} \braket{\Phi_\parallel | R |\Phi_\parallel} 
          - \braket{\Phi_\parallel | r \otimes r |\Phi_\parallel}
                                      \nonumber \\
    & \le \left( \frac{2}{N}\lambda_{\max} - \lambda_{\min}^2 \right) 
              \braket{\Phi_\parallel | \Phi_\parallel} \le 0, 
\end{align}
where we used the fact that the maximum eigenvalue of $R$ never exceeds 
that of $r$, $\lambda_{\max}$, which is proved in the Appendix. 
This completes the proof. 

In the example considered in this section, we find that 
$\lambda_{\max}=\lambda_{\min} = 1/2$ for $N \ge 2$, implying that  
the sufficient condition (\ref{eq_sufficient}) is fulfilled for $N \ge 4$.  
This perfectly agrees with the result obtained previously by detailed calculations.  

\begingroup
\squeezetable
\begin{table}
\caption{The relation of the comparison success probabilities. 
$P_{\circ}^{\rm opt}$ is the optimal success probability. 
$P_{\circ}^{\rm disc}$ is the one based on the discrimination strategy, 
and the result of the no-measurement strategy ($E_{=}=0,\ E_{\ne}=\mbold{1}$) 
is denoted by $P_{\circ}^{\rm no}$.
 }
\begin{ruledtabular}
\begin{tabular}{lll}
$N=2$ & $N=3$ & $N \ge 4$ \\ 
\hline
$P_{\circ}^{\rm no} <  P_{\circ}^{\rm disc} = P_{\circ}^{\rm opt}$  &
$P_{\circ}^{\rm no} = P_{\circ}^{\rm disc} < P_{\circ}^{\rm opt}$  & 
$P_{\circ}^{\rm disc} <  P_{\circ}^{\rm no} = P_{\circ}^{\rm opt}$  \\  
\end{tabular}
\end{ruledtabular}
\label{tab_P} 
\end{table}
\endgroup

\section{Two-state comparison with error margin}
\label{sec_error_margin}
In this section we introduce an error margin in the problem of state comparison. 
If the error margin is sufficiently large, this scheme is reduced to the minimum-error 
comparison studied in the preceding sections. If the error margin is set 0, which implies 
that no error is allowed, the scheme is just unambiguous comparison.   

We consider two pure states $\ket{\phi_1}$ and $\ket{\phi_2}$ with 
equal prior probabilities. 
We first summarize the results of 
the discrimination problem between these states with an error margin $\mu$. 
The POVM now consists of  $\{e_1,e_2,e_?\}$ with $e_?$ associated with the 
inconclusive result.  
The task is to maximize the discrimination success probability
\begin{align}
   Q_{\circ} = \frac{1}{2} \braket{\phi_1 | e_1 | \phi_1} 
                             + \frac{1}{2} \braket{\phi_2 | e_2 | \phi_2},
\end{align}
subject to the condition that the error probability does not exceed an error margin $\mu$
\begin{align}
   Q_{\times} =  \frac{1}{2} \braket{\phi_1 | e_2 | \phi_1} 
                             + \frac{1}{2} \braket{\phi_2 | e_1 | \phi_2} \le \mu, 
\end{align}
and the POVM condition $e_1 + e_2 \le \mbold{1}$. 
This problem was solved in Ref.~\cite{Hayashi08}. The results are
\begin{align}
   Q_{\circ}^{\rm opt} &= \left\{ 
       \begin{array}{ll} 
             \frac{1}{2} \left( 1+\sqrt{1-|\braket{\phi_1 | \phi_2}|^2} \right) 
                             & (\mu_c \le \mu \le 1), \\
             \left( \sqrt{\mu}+\sqrt{ 1-|\braket{\phi_1 | \phi_2 }|} \right)^2 
                            & (0 \le \mu \le \mu_c), \\
       \end{array}                \right.                         \\
    Q_{\times}^{\rm opt} &= \left\{ 
       \begin{array}{ll}
              \mu_c & (\mu_c \le \mu \le 1), \\
              \mu    & (0 \le \mu \le \mu_c),  \\
      \end{array}                \right. 
\end{align}
where 
\begin{align}
   \mu_c = \frac{1}{2} \left(  1-\sqrt{1-|\braket{\phi_1 | \phi_2}|^2} \right). 
\end{align}
For the case of general prior probabilities, see  Ref.~\cite{Sugimoto09}.

In what follows we consider the comparison problem of $\ket{\phi_1}$ and $\ket{\phi_2}$ 
with equal prior probabilities. The POVM now includes an element $E_?$ for the inconclusive 
result in addition to  $E_{=}$ and $E_{\ne}$.  
In the discrimination strategy, the POVM takes the form
\begin{align}
    E_{=}^{\rm disc}     &= e_1 \otimes e_1 + e_2 \otimes e_2, \\
    E_{\ne}^{\rm disc} &= e_1 \otimes e_2 + e_2 \otimes e_1, \\
    E_{?}^{\rm disc}     &= 1- E_{=}^{\rm disc}- E_{\ne}^{\rm disc}.
\end{align}
In this strategy, the comparison is successful if and only if
the two discrimination inferences for subsystems are either both correct or both wrong. 
Therefore the comparison success probability $P_{\circ}$ is given by
\begin{align}
    & P_{\circ}^{\rm disc} = \left(Q_{\circ}^{\rm opt}\right)^2
                                               +\left(Q_{\times}^{\rm opt}\right)^2
                                                 \\
     & = \left\{ 
                  \begin{array}{ll}
                        1-\frac{1}{2}|\braket{\phi_1 | \phi_2}|^2 \ 
                                                                   (\mu_c \le \mu \le 1), \\
                        \left( \sqrt{\mu}+ \sqrt{1-|\braket{\phi_1 | \phi_2}|} \right)^4
                        + \mu^2 \ (0 \le \mu \le \mu_c).  \\
                 \end{array}
            \right.             \nonumber
\end{align}
The comparison in this scheme produces a wrong outcome if one discrimination is 
correct whereas the other is wrong. Thus we obtain 
\begin{align}
    & P_{\times}^{\rm disc} = 2Q_{\circ}^{\rm opt}Q_{\times}^{\rm opt} 
                                \nonumber \\
            &= \left\{ \begin{array}{ll}
                                   \frac{1}{2}|\braket{\phi_1 | \phi_2}|^2 \ (\mu_c \le \mu \le 1), \\
                                   2\mu\left( \sqrt{\mu}+\sqrt{1-|\braket{\phi_1 | \phi_2}|} \right)^2
                                        \  (0 \le \mu \le \mu_c). \\
                              \end{array}
                  \right.  
\end{align}
We rewrite those results in terms of the margin $m$ for the erroneous 
comparison probability; that is, under the error margin condition given by 
\begin{align}
    P_{\times} \le m, 
\end{align} 
the discrimination strategy gives the following results:
\begin{align}
   & P_{\circ}^{\rm disc}  \nonumber \\ 
   &=    \left\{ 
                  \begin{array}{ll}
                        1-\frac{1}{2}|\braket{\phi_1 | \phi_2}|^2 \ 
                                                                   (m_c \le m \le 1), \\
                        \left( \sqrt{\mu}+ \sqrt{1-|\braket{\phi_1 | \phi_2}|} \right)^4
                        + \mu^2 \ (0 \le m \le m_c),  \\
                 \end{array}
            \right.               \\
   & P_{\times}^{\rm disc} \nonumber \\
   &=   \left\{
                  \begin{array}{ll}
                        m_c  \ (m_c \le m \le 1), \\
                        2\mu\left( \sqrt{\mu}+\sqrt{1-|\braket{\phi_1 | \phi_2}|} \right)^2
                                       \ (0 \le m \le m_c), \\ 
                  \end{array}
          \right. 
\end{align}
where the critical error margin $m_c$ is defined by
\begin{align}
      m_c =  \frac{1}{2}|\braket{\phi_1 | \phi_2}|^2, 
                     \label{eq_mc}
\end{align}
and the margin $\mu$ in the discrimination process is related to the margin $m$  
in the comparison in the following way:
\begin{align}
      2\mu\left( \sqrt{\mu}+\sqrt{1-|\braket{\phi_1 | \phi_2}|} \right)^2 = m.
\end{align}

Now we will determine the optimal comparison success probability with an error margin $m$. 
The task is to maximize the success probability
\begin{align} 
    P_{\circ} = \frac{1}{2}\tr E_{=} \rho_{=} +\frac{1}{2}\tr E_{\ne}\rho_{\ne},
\end{align}
subject to the error margin condition
\begin{align}
    P_{\times} = \frac{1}{2}\tr E_{\ne} \rho_{=} +  \frac{1}{2}\tr E_{=}\rho_{\ne} \le m, 
\end{align}
and the POVM conditions
\begin{align}
   & E_{=}, E_{\ne} \ge 0,\\
   &  E_{=} +E_{\ne} \le \mbold{1},  \label{eq_POVMcondition}
\end{align}
where $\rho_{=}$ and $\rho_{\ne}$ are density operators defined to be  
\begin{align}
    \rho_{=} &=  \frac{1}{2}\sum_{k=1}^2 \ket{\phi_k,\phi_k}\bra{\phi_k,\phi_k}, 
                         \\
    \rho_{\ne} &= \frac{1}{2}\sum_{k,j=1(k \ne j)}^2 
                                     \ket{\phi_k,\phi_j}\bra{\phi_j,\phi_k}.
\end{align}

This is a discrimination problem with an error margin between two 
mixed states, which is generally hard to treat analytically.  
However, we can obtain analytical  results 
by using two useful exchange-type symmetries in the problem. 

The first symmetry we consider is concerned with the system swap operation $\Pi$, whose 
action is $\Pi \ket{\phi} \otimes \ket{\psi} = \ket{\psi} \otimes \ket{\phi}$ 
for any state $\ket{\phi}$ and $\ket{\psi}$. 
The second symmetry is a sort of state exchange symmetry. 
Suppose we choose phases of the states so that $\braket{\phi_1 | \phi_2}$ 
is real. Then there exists a unitary $U$ such that 
$U\ket{\phi_1}=\ket{\phi_2}$ and $U\ket{\phi_2}=\ket{\phi_1}$. 
The state exchange operator $\Gamma$ is defined to be $\Gamma = U \otimes U$. 
It is clear that $\rho_{=}$ and $\rho_{\ne}$ are invariant under these two operations.
\begin{align}
   & \Pi \rho_{=} \Pi^\dagger = \rho_{=}
                  ,\ \Gamma \rho_{=} \Gamma^\dagger = \rho_{=},
                           \\
   & \Pi \rho_{\ne} \Pi^\dagger = \rho_{\ne}
                          ,\ \Gamma \rho_{\ne} \Gamma^\dagger = \rho_{\ne}. 
\end{align}

Therefore, the optimal success probability is achieved by a POVM which is invariant 
under the system swap and state exchange operations.
\begin{align}
   & \Pi E_{=} \Pi^\dagger = E_{=}
                  ,\ \Gamma E_{=} \Gamma^\dagger = E_{=},
                        \label{eq_POVMinvariance1}   \\
   & \Pi E_{\ne} \Pi^\dagger = E_{\ne}
                          ,\ \Gamma E_{\ne} \Gamma^\dagger = E_{\ne}.
                      \label{eq_POVMinvariance2}
\end{align}

The space $V = \mathbb{C}^2 \otimes \mathbb{C}^2$ can be decomposed 
into three orthogonal subspaces according to the symmetries with respect to $\Pi$ and 
$\Gamma$.
 \begin{align}
    V = V_{++} \oplus V_{+-} \oplus V_{--}, 
\end{align}
where $V_{\pi,\gamma}$ is the eigenspace in which the eigenvalue of $\Pi$ is 
$\pi$ and the eigenvalue of $\Gamma$ is $\gamma$. 
Note that $\pi=-1$ immediately implies $\gamma=-1$.
The space $V_{++}$ is two-dimensional and spanned by 
\begin{align}
      \ket{X_1} &= \ket{\phi_1,\phi_1}+\ket{\phi_2,\phi_2}, \\
      \ket{X_2} &= \ket{\phi_1,\phi_2}+\ket{\phi_2,\phi_1}, 
\end{align}
which are not orthogonal in general. 
The spaces $V_{+-}$ and $V_{--}$ are both one-dimensional, and 
consist of scalar multiples of 
\begin{align}
    \ket{Y_{+}} = \ket{\phi_1,\phi_1}-\ket{\phi_2,\phi_2}, 
\end{align}
and 
\begin{align}
    \ket{Y_{-}} = \ket{\phi_1,\phi_2}-\ket{\phi_2,\phi_1}, 
\end{align}
respectively.

By the invariance of POVM, Eqs.~(\ref{eq_POVMinvariance1}) and (\ref{eq_POVMinvariance2}), we can write 
\begin{align}
   E_{=} &= \tilde  E_{=} + \alpha_{+}\ket{Y_{+}}\bra{Y_{+}}
                                        + \alpha_{-}\ket{Y_{-}}\bra{Y_{-}}, 
                                            \\
   E_{\ne} &= \tilde  E_{\ne} + \beta_{+}\ket{Y_{+}}\bra{Y_{+}}
                                        + \beta_{-}\ket{Y_{-}}\bra{Y_{-}}, 
\end{align}
where $\alpha_{\pm},\ \beta_{\pm} \ge 0$ and $\tilde  E_{=}$ and $\tilde  E_{\ne}$ 
are semipositive definite operators on the subspace $V_{++}$. 
Then $P_{\circ}$ and $P_{\times}$ read 
\begin{align}
   P_{\circ} = \frac{1}{8} \Big( &
                  \braket{X_1|\tilde E_{=}|X_1}+ \braket{X_2|\tilde E_{\ne}|X_2} 
                                                              \nonumber \\
               & + \alpha_{+}|\braket{Y_+|Y_+}|^2+  \beta_{-}|\braket{Y_-|Y_-}|^2
                                           \Big),                      \\
   P_{\times} = \frac{1}{8} \Big( &
                  \braket{X_1|\tilde E_{\ne}|X_1}+ \braket{X_2|\tilde E_{=}|X_2} 
                                                              \nonumber \\
               & + \beta_{+}|\braket{Y_+|Y_+}|^2+  \alpha_{-}|\braket{Y_-|Y_-}|^2
                                           \Big).
\end{align}

The POVM condition Eq.~(\ref{eq_POVMcondition}) can also be decomposed into 
the conditions in subspaces. In  $V_{+-}$, it is given by 
\begin{align}
     (\alpha_+ + \beta_+)\ket{Y_+}\bra{Y_+} \le \mbold{1}.
\end{align}
In order to maximize $P_{\circ}$ subject to the condition $P_{\times} \le m$, 
it is clear that we should have 
\begin{align}
    \alpha_+^{\rm opt} = \frac{1}{\braket{Y_+|Y_+}},\ \beta_+^{\rm opt}=0.
\end{align}
Similarly, from the POVM condition in $V_{--}$, we obtain 
\begin{align}
       \alpha_-^{\rm opt} = 0,\ \beta_-^{\rm opt}= \frac{1}{\braket{Y_-|Y_-}}.
\end{align}

The POVM operators $\tilde  E_{=}$ and $\tilde  E_{\ne}$ are yet to be 
determined. This should be carried out by maximizing
\begin{align}
 P_{\circ} =& \frac{1+|\braket{\phi_1|\phi_2}|^2}{4}
         \left(\braket{\tilde X_1|\tilde E_{=}|\tilde X_1}
                   + \braket{\tilde X_2|\tilde E_{\ne}|\tilde X_2} \right)
                            \nonumber \\
      &+ \frac{1-|\braket{\phi_1|\phi_2}|^2}{2},   
\end{align}
subject to the error-margin condition given by
\begin{align}
   P_{\times} = & \frac{1+|\braket{\phi_1|\phi_2}|^2}{4}
         \left(\braket{\tilde X_1|\tilde E_{\ne}|\tilde X_1}
                   + \braket{\tilde X_2|\tilde E_{=}|\tilde X_2} \right)
                         \nonumber \\
                      < & m,
\end{align}
where $\ket{\tilde X_k}$ is the normalized state of $\ket{X_k}$  ($k=1,2$).
The POVM condition in $V_{++}$ is 
\begin{align}
    \tilde E_{=} + \tilde E_{\ne} \le \mbold{1}.
\end{align}
The problem in this form is, up to additive and multiplicative constants, equivalent to a discrimination problem with an error margin between 
two pure states $\ket{\tilde X_1}$ and $\ket{\tilde X_2}$ with equal prior probabilities. 
We can employ the optimal solution summarized at the beginning of this section. 
Note that the error margin in this equivalent discrimination problem should be taken as 
$2m/(1+|\braket{\phi_1|\phi_2}|^2)$, and the inner product of the states 
to be discriminated is given by
\begin{align}
\braket{\tilde X_1|\tilde X_2}=  \frac{2\braket{\phi_1|\phi_2}}
                                                              {1+|\braket{\phi_1|\phi_2}|^2}.
\end{align}

Thus we finally obtain the optimal comparison success probability $P_{\circ}^{\rm opt}$ 
with an error margin $m$. 
\begin{align}
  &P_{\circ}^{\rm opt}  \nonumber  \\
  &= \left\{ 
      \begin{array}{l}
            1-\frac{1}{2}|\braket{\phi_1|\phi_2}|^2 \ (m_c \le m \le 1), \\
            \frac{1}{2}+\frac{1}{2}\left( \sqrt{2m}+1\right) 
                         \left( \sqrt{2m}+1-2|\braket{\phi_1|\phi_2}| \right) \\
             \hspace{10em}             (0 \le m \le m_c), \\
      \end{array}
                                        \right.
\end{align}
where the critical error margin is given by $m_c=|\braket{\phi_1 | \phi_2}|^2/2$, 
which is the same as the one in the discrimination strategy.  When $m$ is greater than 
$m_c$, the optimal success probability is given by that of the minimum error scheme  
Eq.~(\ref{eq_minimumerror}).  The result of unambiguous comparison can 
be obtained be setting $m=0$.  We obtain 
\begin{align}
    P_{\circ}^{\rm opt}(\text{unamb}) = 1-|\braket{\phi_1|\phi_2}|, 
\end{align}
which agrees with the result of Ref.~\cite{Barnett03}.  
This probability happens to equal the unambiguous discrimination probability 
$Q_{\circ}^{\rm opt}(\text{unamb})$. Note that the discrimination strategy 
gives $P_{\circ}^{\rm disc}(\text{unamb}) = ( Q_{\circ}^{\rm opt}(\text{unamb}))^2$, 
which is strictly less than the optimal success probability $P_{\circ}^{\rm opt}(\text{unamb})$. 
Thus, as emphasized in Ref.~\cite{Barnett03}, the unambiguous state comparison is not 
reduced to discrimination task for subsystems. 
What we have found is that this is also true for a general error margin as long as the 
error-margin condition is active; namely, we can show that 
$ P_{\circ}^{\rm opt} > P_{\circ}^{\rm disc}$ for $0 \le m < m_c$. 
In Fig.~\ref{fig_1} we display  $P_{\circ}^{\rm opt}$ and  $P_{\circ}^{\rm disc}$ as 
functions of error margin $m$ in the case of  $|\braket{\phi_1|\phi_2}|=0.8$. 

\begin{figure} 
\includegraphics[width=9cm]{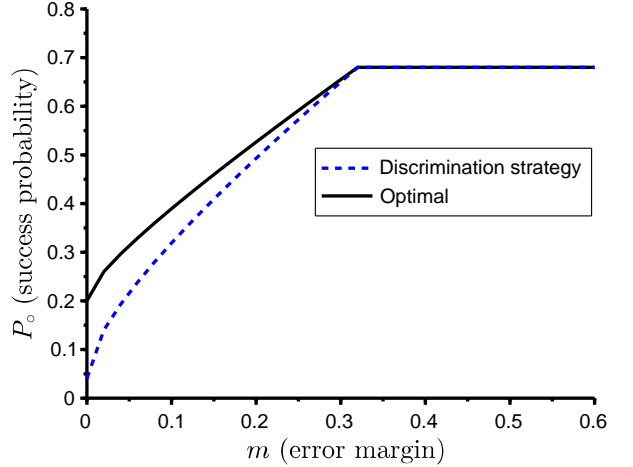}
\caption{(Color online) The comparison success probability vs error margin. 
The solid line is the optimal success probability $P_{\circ}^{\rm opt}$. 
The dashed line is the success probability $P_{\circ}^{\rm disc}$ obtained by 
the discrimination strategy. In this example, we assume $|\braket{\phi_1|\phi_2}|=0.8$, 
implying $m_c = 0.32$.}
\label{fig_1}
\end{figure}

\section{Summary and Concluding remarks}
\label{sec_concluding}
The aim of this paper was to investigate the performance of the discrimination strategy in 
the comparison task of known quantum states. In the case of minimum-error comparison of 
two states, the discrimination strategy always gives the optimal result. 
We have shown that this result can be generalized for mixed states with arbitrary prior probabilities. However, if the number of possible states is greater than two, this is no longer 
true. In some cases, the optimal comparison strategy is simply that we infer the two 
systems are in different states without any measurement. 
Rather surprisingly the discrimination strategy performs worse than this 
``no-measurement'' strategy.
A sufficient condition for this phenomenon to occur was presented.

We have also investigated how the constraint on the error probability (error margin) 
affects the performance of the discrimination strategy, and found that 
the discrimination strategy is optimal only in the minimum-error case.
In the case where the error-margin condition is active (including the unambiguous scheme),  
the optimal comparison requires collective measurement on the whole system. 

In this paper we assumed that the states are selected  from a known set of states. 
We can consider a different scheme; that is, we have no classical knowledge of the 
possible states but some number of copies of the states are available instead 
\cite{Hayashi05,Bergou05,Hayashi06}. 
The corresponding discrimination task is sometimes called quantum state identification. 
It will be of interest in future studies to extend our investigation to the 
``identification'' strategy in the state comparison in this setting.

\appendix*
\section{Maximum eigenvalues of $R$ and $r$}
We prove the following theorem which is used in the end of Sec.~\ref{sec_many}:
\begin{theorem}
Let $\{\ket{\phi_0^A},\ket{\phi_1^A},\ldots,\ket{\phi_{N-1}^A}\}$ and
$\{\ket{\psi_0^B},\ket{\psi_1^B},\ldots,\ket{\psi_{N-1}^B}\}$ be sets of 
$N$ normalized pure states of system A and B, respectively. 
Define three density operators $r_{\phi}^A$, $r_{\psi}^B$, and $R^{AB}$. 
\begin{align*}
   r_{\phi}^A & \equiv  \frac{1}{N}\sum_{k=0}^{N-1} \ket{\phi_k^A}\bra{\phi_k^A},
\end{align*}
\begin{align*}
   r_{\psi}^B & \equiv  \frac{1}{N}\sum_{k=0}^{N-1} \ket{\psi_k^B}\bra{\psi_k^B}, 
\end{align*}
and
\begin{align*}
   R^{AB} & \equiv \frac{1}{N} \sum_{k=0}^{N-1}
                               \ket{\phi_k^A,\psi_k^B}\bra{\phi_k^A,\psi_k^B}.
\end{align*}
Then we have 
\begin{align}
   \lambda_{\max}(R^{AB}) 
     \le \min \left\{ \lambda_{\max}(r_{\phi}^A), \lambda_{\max}(r_{\psi}^B) \right\},
\end{align}
where $\lambda_{\max}(\Omega)$ stands for the maximum eigenvalue of an 
operator $\Omega$.
\end{theorem}

\begin{proof}
For any normalized state $\ket{\Phi^{AB}}$ of system AB, we observe
\begin{align*}
  & \braket{\Phi^{AB}| R^{AB} |\Phi^{AB} } \\
  & = \frac{1}{N} \sum_{k=0}^{N-1}
           \braket{ \Phi^{AB} \ket{\phi_k^A,\psi_k^B}\bra{\phi_k^A,\psi_k^B} 
                                                              \Phi^{AB} } \\
  & \le \frac{1}{N} \sum_{k=0}^{N-1}
           \braket{ \Phi^{AB}| 
              \left(  \ket{\phi_k^A}\bra{\phi_k^A}  \otimes \mbold{1}^B  \right) 
                                                              | \Phi^{AB} } \\
  & = {\rm tr}_A 
     \left( \frac{1}{N} \sum_{k=0}^{N-1}\ket{\phi_k^A}\bra{\phi_k^A} \right) \rho^A \\
  & \le \lambda_{\max}(r_{\phi}^A), 
\end{align*}
where $\rho^A$ is the reduced density operator defined by 
\begin{align*}
  \rho^A = {\rm tr}_B \ket{\Phi^{AB}}\bra{\Phi^{AB}}.
\end{align*} 
Since 
\begin{align*}
   \lambda_{\max}(R^{AB}) = \max_{|\Phi^{AB}|=1}
                           \braket{\Phi^{AB}| R^{AB} |\Phi^{AB} }, 
\end{align*}
we obtain $\lambda_{\max}(R^{AB}) \le \lambda_{\max}(r_{\phi}^A)$. 
Similarly we can show that $\lambda_{\max}(R^{AB}) \le \lambda_{\max}(r_{\psi}^B)$. 
Combining these two results we obtain the desired inequality.
\end{proof}


\end{document}